\documentclass[runningheads]{llncs}
\pdfoutput=1
\usepackage{graphicx}
\usepackage{hyperref}
\usepackage{amssymb}

\def\cut{\not\rightarrow}

\makeatletter
\def\hb@xt@{\hbox to }
\makeatother

\let\oldendproof\endproof
\def\endproof{\qed\oldendproof}

\title{Flows in One-Crossing-Minor-Free Graphs}

\author{Erin Chambers\inst{1} \and David Eppstein\inst{2}}
\institute{\noindent Dept. of Mathematics and Computer Science, Saint Louis University
\and Computer Science Department, University of California, Irvine}

\date{ }

\begin{document}

\maketitle

\begin{abstract}
We study the maximum flow problem in directed $H$-minor-free graphs where $H$ can be drawn in the plane with one crossing. If a structural decomposition of the graph as a clique-sum of planar graphs and graphs of constant complexity is given, we show that a maximum flow can be computed in $O(n\log n)$ time. In particular, maximum flows in directed $K_{3,3}$-minor-free graphs and directed $K_5$-minor-free graphs can be computed in $O(n\log n)$ time without additional assumptions.
\end{abstract}


\section{Introduction}

Computing maximum flows is fundamental in algorithmic graph theory, and has many applications. Although flows can be computed in polynomial time for arbitrary graphs, it is of interest to find classes of graphs for which flows can be computed more quickly, and specialized algorithms are known for flows in planar graphs~\cite{bk-amfdp-09,ff-mfn-56,h-mfpn-81,hj-oamfu-85,is-mfpn-79,jv-udcffdpn-82,kn-fpgwvc-94,knk-lsfpg-93,mn-fpgms-95,w-mstfpn-97}, graphs of bounded genus~\cite{cen-hfcc-09,cen-mcshc-09}, graphs with small crossing number~\cite{hw-mstfkc-07}, and graphs of bounded treewidth~\cite{hknr-cmfnc-98}.

Planar graphs, graphs of bounded genus, and graphs of bounded treewidth are \emph{minor-closed graph families}, families of graphs closed under  edge contractions and edge deletions.  According to the Robertson--Seymour graph minor theorem~\cite{rs-gm20wc-04}, any minor-closed graph family can be described as the $X$-minor-free graphs, graphs that do not have as a minor any member of a finite set $X$ of non-members of the family; for instance, the planar graphs are exactly the $\{K_5,K_{3,3}\}$-minor-free graphs~\cite{wagnerk5}. In many cases the properties of a graph family are closely related to the properties of its excluded minors: for instance, the minor-closed graph families with bounded treewidth are exactly the families of $X$-minor-free graphs for which $X$ includes at least one planar graph~\cite{rs-gm5epg-86}, and the families with bounded local treewidth (a functional relationship between the diameter of a graph and its treewidth) are exactly those for which $X$ includes at least one \emph{apex graph}, a graph that can be made planar by removing a single vertex~\cite{e-dtmcgf-00}. If $X$ includes a graph that can be drawn in the plane with a single pair of crossing edges, then the $X$-minor-free graphs have a structural decomposition as a \emph{clique-sum} of smaller graphs that are either planar or have bounded treewidth~\cite{Demaine021.5-approximationfor,rs-onecrossing}. In this last case we say that the family of $X$-minor-free graphs is \emph{one-crossing-minor-free}; families of this type include the $K_{3,3}$-minor-free graphs and $K_5$-minor-free graphs, since $K_{3,3}$ and $K_5$ are one-crossing graphs (Figure~\ref{fig:onecrossing}).

\begin{figure}[t]
\centering\includegraphics[width=2.75in]{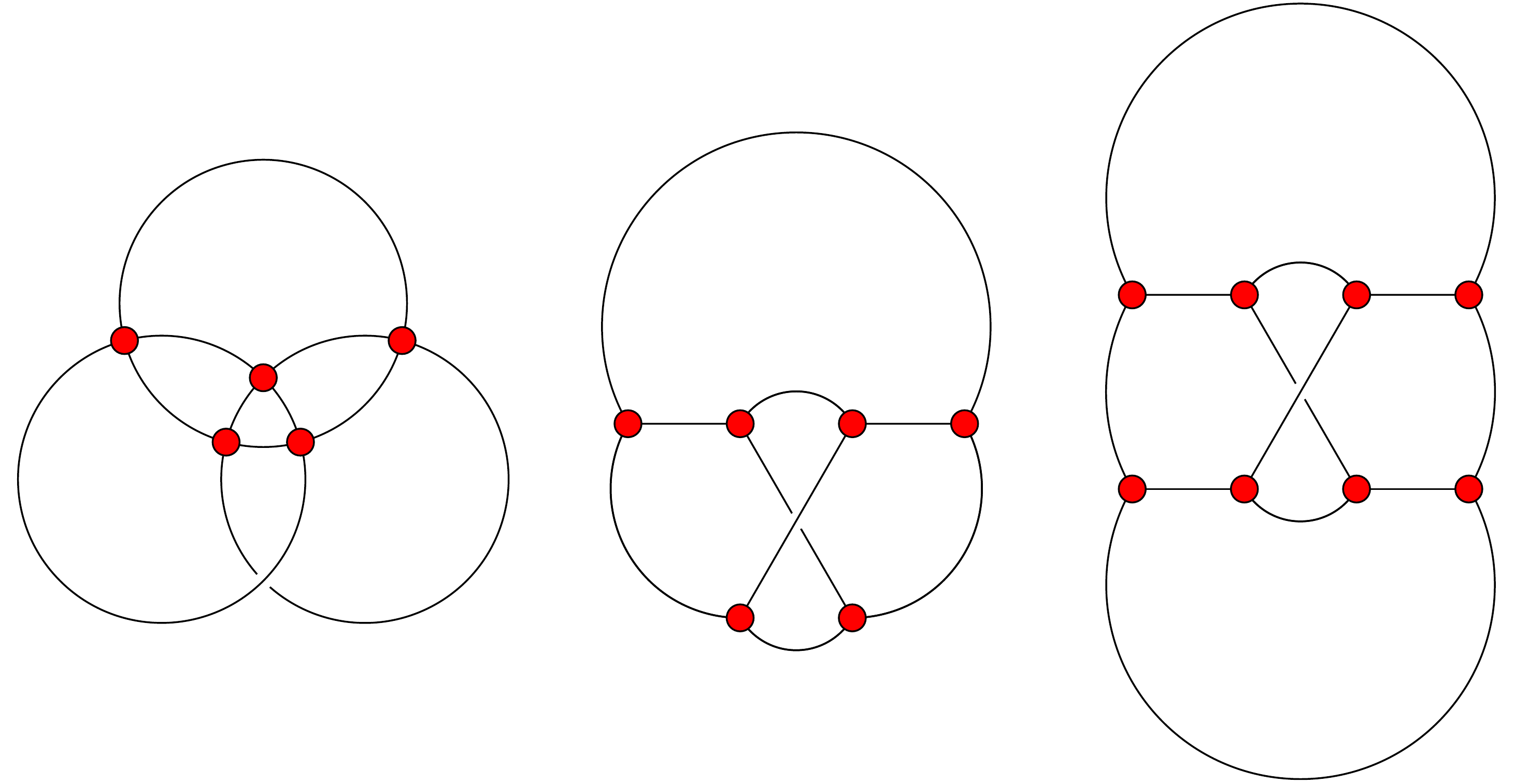}
\caption{One-crossing drawings of $K_5$ (left), $K_{3,3}$ (center), and the Wagner graph (right).}
\label{fig:onecrossing}
\end{figure}

In this paper we consider flows in one-crossing-minor-free graph families. We provide $O(n \log n)$ algorithms to compute maximum flows in any directed $H$-minor-free graph, where $H$ is a fixed one-crossing graph and where the structural decomposition of the graph is provided as part of the input. In the case of $K_{3,3}$-minor-free graphs and $K_5$-minor-free graphs, algorithms are known that can find such a decomposition efficiently~\cite{asano85,rl-ork5mfg-08}, and by combining our techniques with those known decomposition algorithms we provide an $O(n\log n)$ time algorithm for maximum flow in directed $K_{3,3}$-minor-free and $K_5$-minor-free graphs without requiring the decomposition to be part of the input.

Our main motivation for looking at flows in one-crossing-minor-free graphs is to try to make progress towards finding flows in arbitrary minor-closed graph families. Due to known separator theorems for minor-closed families~\cite{ast-stgema-90}, an algorithm of Johnson and Venkatesan~\cite{jv-udcffdpn-82} can be applied to find flows in any minor-closed family in time $O(n^{3/2}\log n)$, but this does not come close to the nearly-linear time bounds known for planar and bounded-genus graphs.
Like one-crossing-minor-free graphs, graphs in more general minor-closed families have a structural decomposition in terms of bounded-genus surfaces, clique-sums, apexes (a constant number of vertices that can be adjacent to arbitrary subsets of each of the bounded-genus surfaces), and vortexes (bounded-treewidth graphs glued into faces of the bounded-genus surfaces)~\cite{dhk-agmtdac-05,rs-gm16enpg-03}. However, these decompositions are greatly simplified in the one-crossing-minor-free case: the surfaces are planes and there are no apexes or vortexes. To handle the general case, we would need to combine clique-sums, bounded-genus, apexes, and vortexes.  The problem of flows on bounded genus surfaces has been previously examined~\cite{cen-hfcc-09,cen-mcshc-09} and the present work focuses on clique-sums, as these are the main feature in the structural decomposition for one-crossing-minor-free graphs. However, it remains unclear how to handle apexes and vortexes.

As an important tool in our results, we greatly simplify the \emph{mimicking networks} of Hagerup et al~\cite{hknr-cmfnc-98} for multiterminal flow networks in the case of four terminals with a single source, leading to significantly reduced constant factors in the running time of algorithms that use this networks. Similar simplifications had been achieved in the undirected case~\cite{cswz-cmn-00} but to our knowledge our small directed mimicking network is novel.

 \section{Preliminaries}

\subsection{Flows and Cuts}

A \emph{flow network} is a graph $G=(V,E)$, where each edge $e \in E$ has an associated nonnegative real capacity $c_e$, along with two distinguished vertices $s$ and $t$ which are called the \emph{source} and \emph{sink}, respectively.  A \emph{flow} in this graph is a set of nonnegative real values $f_e$ for each edge $e \in E$ such that $f_e \le c_e$ for every $e \in E$ and $\sum_{uv \in E} f_{uv} = \sum_{(vu) \in E} f_{vu}$ for every $v \in V - \{s,t\}$.  The \emph{value} of the flow is the amount of net flow going from $s$ to $t$, or $\sum_{sv \in E} f_{sv}$. 

A \emph{cut} in a flow network is a set of edges separating $s$ from $t$; the \emph{capacity} of the cut is the sum of the capacities of all the edges that cross the cut in the direction from $s$ to $t$.  Clearly, the value of any flow is less than or equal to the value of any cut.  The classic max-flow min-cut theorem states that the maximum possible flow from $s$ to $t$ is in fact equal to the minimum possible cut.

In our setting, we will need to know the maximum possible flow that can travel through a subgraph of the input graph; this will allow us to simplify the graph by removing the subgraph and replacing it with an equivalent (but much smaller) subgraph.  To do this, we compute an \emph{external flow}.  In external flow networks, instead of a single source and sink, we have an ordered set of terminals $Q = \{ q_1, \ldots q_k \}$ where each terminal $q_i$ has associated with it a number $x_i$; we require that $\sum_i x_i=0$. If $x_i$ is positive then it is interpreted as a supply, and if $x_i$ is negative it is interpreted as a demand. (It may also be the case that $x_i$ is zero, in which case $x_i$ carries neither a supply nor a demand.)  A \emph{realizable external flow} is a set of $k$ values $(x_1, \ldots, x_k)$ for $(q_1, \ldots, q_k)$, along with a flow $f$ such that $\sum_{(q_i,v) \in E} f_{(q_i,v)} - \sum_{(v,q_i) \in E} f_{(v,q_i)} = x_i$ for all $i$.  Basically, the flow remains balanced at every vertex in $V \setminus Q$, and the imbalance of flow at each vertex in $q_i \in Q$ is exactly $x_i$.

It will be helpful to define a special case of external flow networks, which we call \emph{single-source external flow networks}. A single-source external flow network is, simply, an external flow network for which only $q_1$ may have a positive supply $x_i$; every other terminal has $x_i\le 0$ indicating that it is either a demand node or is inactive as a terminal. As before, a \emph{realizable single-source external flow} is a realizable external flow subject to this constraint on the values of $x_i$.

We define $S\cut T$, for sets of terminals $S$ and $T$ in an external flow network, to be the minimum value of a cut for which every terminal in $S$ is on the source side of the cut and every terminal in $T$ is on the sink side of the cut. We will further abbreviate the notation by writing strings of symbols instead of bracketed set notation for the sets of terminals on each side of the cut; e.g., $s\cut abc$ should be interpreted as an abbreviation for $\{s\}\cut\{a,b,c\}$.

\subsection{Mimicking networks}

Let $G$ be an external flow network or single-source external flow network with a fixed specification of the edge capacities and a fixed ordered set of terminals $Q$, but where the supply and demand quantities $x_i$ remain variable. A key ingredient for our technique, as formalized by Hagerup et al.~\cite{hknr-cmfnc-98}, is the concept of a \emph{mimicking network}, a network $H$ that has the same terminals as $G$ and has the same realizable external flows, but may have many fewer edges and vertices than $G$.  In what amounts to a form of separator based sparsification~\cite{egis-sbs-93}, Hagerup et al.{} solve flow problems on bounded-treewidth networks by repeatedly replacing subnetworks of the given network with smaller mimicking networks.
Their construction of mimicking networks is based on an observation of Gale~\cite{gale}:

\begin{lemma} \cite{hknr-cmfnc-98} \label{hagerup}
An external flow $(x_1, \ldots, x_k)$ is realizable in a network $G=(V,E)$ with terminals $Q = \{ q_1, \ldots, q_k \}$ if and only if the following relations are satisfied:
\begin{eqnarray}
\sum_{i=1}^k x_i & = & 0\\
\sum_{q_i \in S} x_i & \le & \bigl(S\cut (Q\setminus S)\bigr), \mbox{ for all } S \subseteq Q \mbox{ with } \emptyset\ne S\ne Q.
\end{eqnarray}
\end{lemma}

Essentially, this means that in order to understand the possible flow patterns in a subnetwork with $k$ terminals, one needs only to know the $(2^k - 2)$ minimum cut values from a nonempty subset of terminals to its nonempty complement. If two networks have the same minimum cut values for each subset then they behave the same with respect to flows.  Based on this observation, Hagerup et al.{} show that there exists a replacement network of at most $2^{2^k-2}$ vertices that behaves equivalently to any $k$-terminal network:

\begin{figure}[t]
\centering\includegraphics[width=\textwidth]{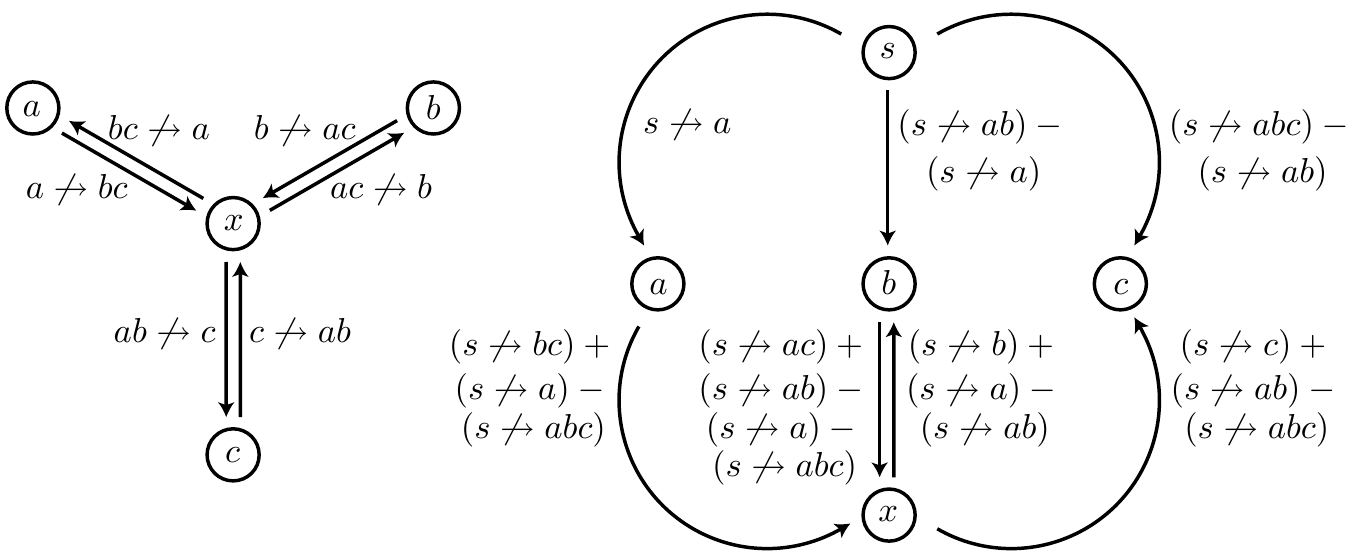}
\caption{Mimicking networks for an external flow network with three terminals $a$, $b$, and $c$ (left), and for a single-source external flow network with four terminals $s$, $a$, $b$, and $c$ (right). In the single-source mimicking network, $a$, $b$, and $c$ are permuted if necessary so that $s\cut a\ge\max \{s\cut b,s\cut c\}$ and $s\cut ab\ge s\cut ac$.}
\label{mimics}
\end{figure}

\begin{lemma} \cite{ hknr-cmfnc-98} \label{hagerup2}
Given any external flow network $G$ with $k$ terminals, there exists a flow network having at most $2^{2^k-2}$ vertices which has the same external flow value as $G$.
\end{lemma}

Specifically, the mimicking network of Lemma~\ref{hagerup2} can be constructed from $G$ by finding a set of $2^k-2$ minimum cuts, one for each partition of the terminals into two nonempty subsets, and by collapsing subsets of vertices in $G$ into a single supervertex whenever all vertices in the subset are on the same side of every cut. Note that the mimicking network is not necessarily a minor of $G$, as the collapsed subsets need not form connected subgraphs of $G$.

However, the size of the mimicking networks formed by Lemma~\ref{hagerup2}, while constant, is large. Our algorithms will involve external flow networks with up to four terminals, and if we applied Lemma~\ref{hagerup2} directly we might get as many as 16384 vertices in our mimicking networks. 
It is possible to reduce the number of vertices in the construction of Hagerup et al. from a power of two to a Dedekind number by requiring nested partitions of the terminals to have nested cuts, but this would still lead to 168 vertices for the mimicking network of a four-terminal network.
We describe in an appendix a simpler mimicking network of Chaudhuri et al.~\cite{cswz-cmn-00}  for external flows with at most three terminals, and a new mimicking network for single source external flows with at most four terminals that are both much smaller, requiring at most one nonterminal vertex. These simplified mimicking networks are depicted in Figure~\ref{mimics}.
It is important for our techniques that the mimicking network for a three-terminal network is planar and has its three terminals in a single face of its planar embedding (more strongly, in fact, it is outerplanar); however, we do not rely on the planarity of the four-terminal mimicking network.

\subsection{Structure of minor free graphs}
 
A \emph{minor} of a graph $G$ is a graph that can be formed from $G$ by contracting and removing edges. A graph family $F$ is \emph{minor-closed} if every minor of a graph in $F$ also belongs to $F$. If $X$ is a finite set of graphs, the $X$-minor-free graphs are the graphs $G$ such that no minor of $G$ belongs to $X$; the $X$-minor-free graphs are obviously a minor-closed graph family, and (much less obviously) the Robertson--Seymour graph minor theorem~\cite{rs-gm20wc-04} states that every minor-closed graph family has this form. If $F$ is the family of $X$-minor-free graphs, then $X$ is the set of \emph{forbidden minors} for $F$; for instance, $K_5$ and $K_{3,3}$ are the forbidden minors for the planar graphs. We will abuse notation and abbreviate the $\{H\}$-minor-free graphs (where $H$ is a single graph) as the $H$-minor-free graphs.

\begin{figure}[t]
\centering\includegraphics[width=3in]{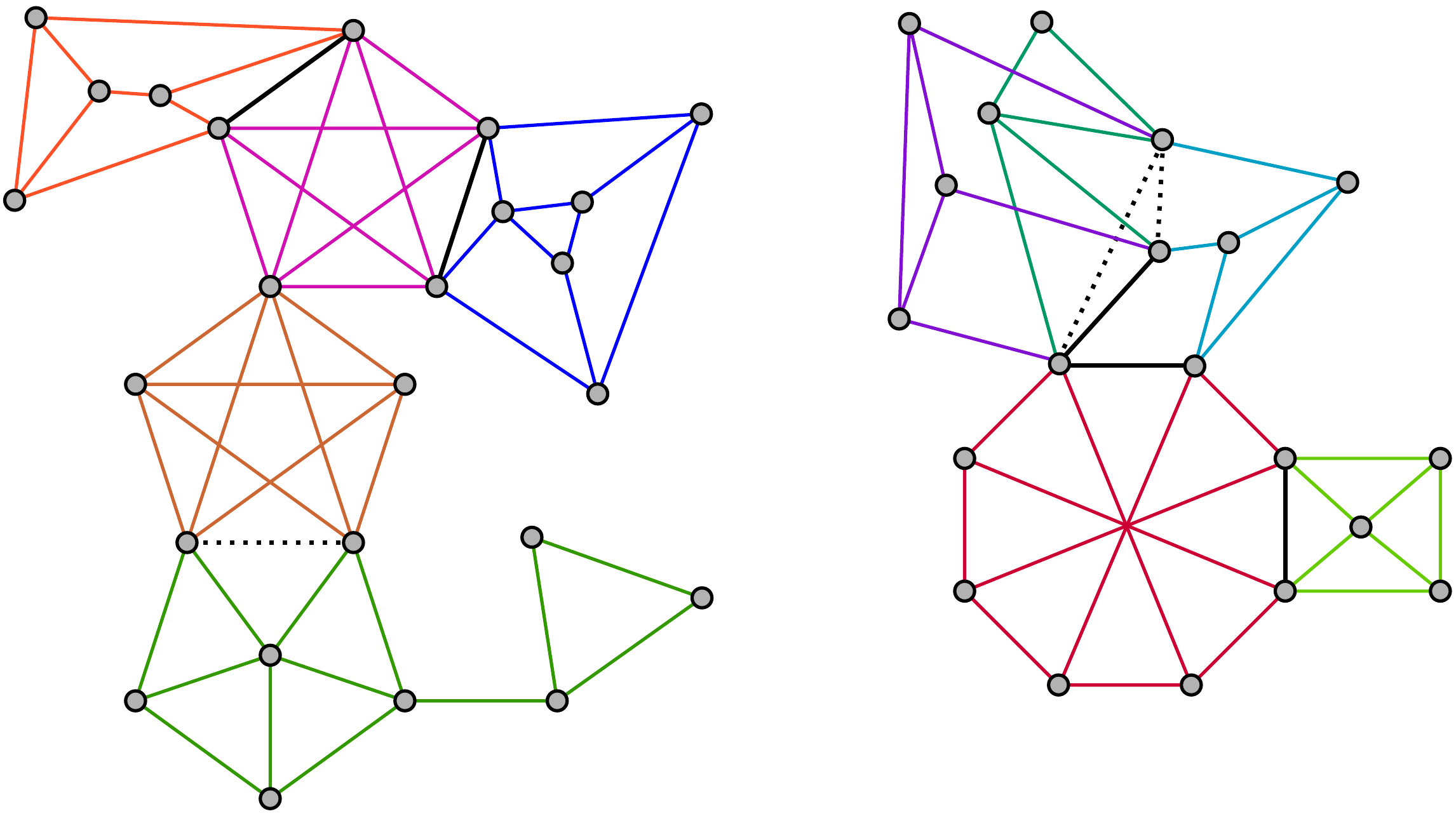}
\caption{A $K_{3,3}$-free graph expressed as a 2-sum of planar graphs and copies of $K_5$ (left), and a $K_5$-free graph expressed as a 3-sum of planar graphs and copies of the Wagner graph (right). The edges forming the cliques of the clique-sum operations are shown as solid black (if they remain in the final graph) or dashed (if they were removed after the clique-sum operation); the colors of the remaining edges identify the subgraphs entering into the clique-sum operations.}
\label{fig:ksums}
\end{figure}

A \emph{clique-sum} is an operation that combines two graphs by identifying the vertices in two equal-sized cliques in the two graphs, and then possibly removing some of the edges of the cliques. A \emph{$k$-sum} is a clique-sum where all cliques have at most $k$ vertices. More generally, we will say that a given graph is a $k$-sum of a collection of more than two graphs if it can be formed by repeatedly replacing pairs of graphs of the collection by their $k$-sum.
Clique-sums are closely related to vertex-connectivity of graphs: for instance, the decomposition of a biconnected graph into triconnected components, formalized in the SPQR-tree, gives a representation of the graph as a $2$-sum of triconnected graphs, cycles, and two-vertex multigraphs~\cite{dt-ipt-89,ht-dgtc-73,m-scpcg-37}. (As some of our algorithms use the SPQR-tree decomposition, we include a brief description of it in an appendix.)

Besides their relation to connectivity, clique-sums have also played an important role in describing the structure of minor-closed graph families since the proof by Wagner \cite{wagnerk5} that $K_5$-minor-free graphs are exactly the graphs that can be formed by 3-sums of planar graphs and the 8-vertex nonplanar Wagner graph, the graph formed by adding four edges connecting opposite pairs of vertices in an 8-cycle (Figure~\ref{fig:ksums}, right). Wagner \cite{wagnerk33} and D. W. Hall \cite{hall} also proved that $K_{3,3}$-minor-free graphs are exactly the 2-sums of planar graphs and the five-vertex complete graph~$K_5$ (Figure~\ref{fig:ksums}, left).
These two characterizations of $H$-minor-closed families can be generalized to any $H$-minor-free graph with the property that $H$ can be drawn in the plane with only a single edge crossing, as $K_{3,3}$ and $K_5$ both can. In this case, the $H$-minor-free graphs can always be decomposed into 3-sums of planar graphs and graphs of bounded treewidth~\cite{rs-onecrossing}.

Algorithmically, a decomposition of a $K_{3,3}$-minor-free graph into a 2-sum of planar graphs and $K_5$ can be computed in linear time~\cite{asano85}. Essentially, this is simply a matter of constructing the SPQR tree and verifying that all triconnected components are either planar or $K_5$.
As has been shown more recently, a decomposition of a $K_5$-minor-free graph into a 3-sum of planar graphs and the Wagner graph can also be constructed in linear time~\cite{139475,rl-ork5mfg-08}.
Algorithmic versions of the generalized decomposition of graphs in any $H$-minor-free family (where $H$ is a one-crossing graph) into 3-sums of planar graphs and graphs of bounded treewidth can also be solved in polynomial time~\cite{Demaine021.5-approximationfor}, but the polynomial $O(n^4)$ is too high to be of use in flow algorithms, so in this case we will assume that a decomposition has been given to us as part of the input. In case a future development leads to linear time algorithms for decomposition of one-crossing-minor-free graphs, this assumption may be removed.

\section{Flow algorithm}

Recall first that any $H$-minor free graph $G$ (where $H$ is a one-crossing graph) can be decomposed into 3-sums of planar and bounded treewidth graphs; we refer to these smaller graphs as \emph{components} of the clique-sum decomposition.  In order to handle the case that three or more components are glued together by a 3-sum at a single shared clique, we may represent the decomposition as a 2-colored tree, in which the vertices of one color represent the components of the decomposition and the vertices of the other color represent the cliques at which they are glued together. We may identify two distinguished components of the decomposition, one containing $s$ and another containing $t$, where $s$ and $t$ are the two terminals of our given flow problem; if $s$ or $t$ is part of a clique on which multiple components are glued, then the distinguished component containing that terminal may be chosen arbitrarily among all components containing it. The two distinguished components are necessarily connected by a path in the clique-sum decomposition tree (possibly a path of length zero).

The vertices of $G$ that belong to the cliques of the clique-sum decomposition, in effect, will form terminals in each component, so any flows into and out of components of the decomposition can be treated as an external flow computations.  Our algorithm will iteratively replace each component in the decomposition with a mimicking network of constant size.  However care is needed in ordering these computations, as some components may have more than a constant number of vertices that belong to cliques of the decomposition and we can only use mimicking networks for components that have a constant number of terminals. We perform this replacement in two phases: in the first phase, the components that are replaced do not belong to the path from $s$ to $t$ of the decomposition tree, and in the second phase the replaced components lie along this path.

\paragraph{\bf Refining the decomposition.}
For technical reasons it is necessary to ensure that the planar components in our decomposition are only glued together by 3-sums along their faces, which we achieve by refining the given clique-sum decomposition tree. To do so, we make the following changes, each of which replaces one of the components in the decomposition by a clique-sum of smaller components:
\begin{itemize}
\item If any component of the decomposition is not biconnected, replace it with the 1-sum of its biconnected components, glued together by 1-sums at its articulation vertices.
\item If any component of the decomposition is biconnected but not triconnected, find an SPQR tree representing it as a 2-sum of its triconnected components~\cite{dt-ipt-89,ht-dgtc-73,m-scpcg-37}, as outlined in an appendix.
\item At this stage of refinement, each component is triconnected, and in particular each planar component has a unique planar embedding. Within each planar component, find all the triangles at which 3-sums are glued, and compare them against the list of triangular faces of the graph. The gluing triangles that are not faces are separating triangles of the component, and may be used to partition each planar component into a 3-sum of smaller planar components.
\end{itemize} 
All these refinements may be performed in linear time, and after this stage each triangle in each planar component of the decomposition forms a face of the component.

\paragraph{\bf Simplification phase I: components off the terminal path.}
Next, we deal with all components that are \emph{not} on the $s$ to $t$ path in the clique-sum decomposition tree.  Our algorithm iteratively finds a component $C_i$ that is a leaf in the clique-sum decomposition tree; let $C_j$ be the component to which $C_i$ is glued by a clique-sum, one step closer to the $s$--$t$ path.
Then $C_i$ is connected to $C_j$ and the rest of the graph via at most three vertices, so we can treat the flow calculation within that component as an external flow problem with at most three terminals.  Since the subgraph is either planar or of bounded treewidth, we can in $O(n \log n)$ time compute the min-cut values of all 6 possible partitions of its terminals, and by Lemma~\ref{mimic3}, we can replace the component $C_i$ with a mimicking network $C_i'$ of at most four vertices. We then split into cases according to how $C_i$ and $C_j$ fit into the rest of the clique-sum decomposition tree:
\begin{itemize}
\item If $C_j$ is a bounded-treewidth component, then $C_i'$ can be glued directly into $C_j$ by performing the clique-sum that connects these two components, forming a decomposition tree with one fewer component. The treewidth of the merged component is the maximum of 3 and the treewidth of $C_j$ (using the fact that the mimicking network for $C_i$ is at worst an oriented form of $K_4$), so after any number of repetitions of steps of this type the treewidth remains bounded.
\item If $C_j$ is a planar component, and $C_i$ connects to $C_j$ via one or two terminals or via a triangle of three terminals that is not used for any other 3-sum in the decomposition tree, then again $C_i'$ can be immediately glued into $C_j$ forming a decomposition tree with one fewer component. Due to the planar structure of our three-terminal mimicking networks and due to the fact that the gluing triangle is a face of $C_j$, this step preserves both the planarity of $C_j$ and the property that all gluing triangles in the new larger planar component are still faces of their component.
\item If $C_j$ is a planar component and the gluing triangle of $C_i$ is shared with another four-vertex replacement network $C_h'$, then $C_h'$ and $C_i'$ may be merged into a single mimicking network. Then, if this merged component is the only component that uses that gluing triangle, it may be glued into $C_j$ as in the previous case.
\item In the remaining cases, we replace $C_i$ by $C_i'$ in the decomposition tree but do not merge it with $C_j$.
\end{itemize}
After all simplifications of this type have been performed, the remaining decomposition tree must have the structure of a path of components connecting $s$ to $t$, where each of the gluing triangles in this path may also be connected to a single mimicking network off the path.

\paragraph{\bf Simplification phase II: components along the terminal path.}
To complete the algorithm we will perform a similar sequence of replacements along the path from $s$ to $t$.  However, we must use a slightly different technique, since we may now need to deal with four terminals at a time: $s$ and the three vertices in the 3-sum connecting the first two components in the $s$--$t$ path.  Additionally, after some of the replacements in this final stage of the algorithm, we will have to compute flows in networks that are neither planar nor of bounded treewidth, but that can be made planar by the removal of a constant number of edges; a result of Hochstein and Weihe \cite{hw-mstfkc-07} allows us to compute minimum cuts in these graphs in near-linear time. In this stage of the algorithm, as long as $s$ and $t$ belong to different components, we perform the following steps:
\begin{itemize}
\item Let $C_i$ be the component containing $s$ and $T_i$ be the set of vertices in the clique-sum connecting $C_i$ to the next component in the $s$--$t$ path.
\item Compute the minimum cut amounts between $s$ and each nonempty subset of $T_i$, using the algorithm of Hochstein and Weihe \cite{hw-mstfkc-07} if $C_i$ has been formed by adding a constant number of edges to a planar component, and the algorithm of Hagerup et al.~\cite{hknr-cmfnc-98} if $C_i$ is a bounded-treewidth component.
\item Replace $C_i$ by a mimicking network $C_i'$ for single-source external flows (as described in Lemma~\ref{mimic4}) using the computed cut amounts.
\item Glue $C_i'$, and (if it exists) the other mimicking network sharing the same gluing triangle, into the next component in the $s$--$t$ path, forming a path with one fewer component. If the next component was planar, it becomes a graph formed from a planar graph by adding a constant number of edges, while if it had bounded treewidth, its new treewidth is again bounded by the maximum of its old treewidth and~3.
\end{itemize}
Eventually this simplification will leave us with a planar or bounded treewidth graph containing both $s$ and $t$, and we can compute the maximum flow between them directly in near linear time.

\paragraph{\bf Reversing the simplifications and constructing a flow.}
We then reverse the sequence of simplifications we have performed, replacing each mimicking network with the larger network it replaced; for each such replacement we perform a single flow computation to find valid flow amounts forming the same external flow as in the mimicking network. At the end of this reversal of the simplification process, we will have a correct maximum flow in the original network that we were given as input.

We summarize our results as a theorem.

\begin{theorem}
If we are given as input a decomposition of a directed flow network into a 3-clique-sum of planar and bounded-treewidth components, then in $O(n\log n)$ time we may compute a maximum flow between any two terminals $s$ and $t$ of the network.
\end{theorem}

\begin{corollary}
For any fixed one-crossing graph $H$, maximum flows in directed $H$-minor-free flow networks may be computed in $O(n\log n)$ time once a clique-sum decomposition of the network has been found.
\end{corollary}

In the case of $K_5$-free and $K_{3,3}$-free graphs, we can find a clique-sum decomposition in $O(n)$ time~\cite{asano85,rl-ork5mfg-08}. The case of $K_{3,3}$-free graphs is particularly simple, since it uses only 2-sums, and therefore involves simpler mimicking networks.

\begin{corollary}
We may compute maximum flows between any two terminals in a $K_5$-free or $K_{3,3}$-free directed flow network in $O(n\log n)$ time.
\end{corollary}


\section{Conclusions}

We have shown how to find flow in near linear time when the input graph is a clique-sum of planar and bounded tree-width graphs by using mimicking networks to iteratively simplify the graph. This technique allows us to use known near linear algorithms in each bounded tree-width or planar component of the decomposition.

There is no added generality in considering 4-sums of planar graphs in place of 3-sums (any 4-sum involving a planar graph can be rearranged into a combination of 3-sums), but our methods immediately generalize to 2-sums of bounded-genus graphs and bounded-treewidth graphs. Flow computation in 3-sums of bounded genus graphs is more problematic due to the possible existence of non-facial non-separating triangles.

The larger goal, however, is computing flow quickly in more arbitrary minor-free families of graphs.  Since flow can be computed efficiently in bounded genus and bounded tree-width graphs, the primary remaining open questions are those of computing flow in graphs with vortices or apexes, since these are the relevant building blocks for more general minor free families.   Even the case of a planar graph with two apexes would be of great interest, since it could be used to solve a generalization of the maximum flow problem on planar graphs in which the flow is allowed to have multiple sources and sinks.

{\raggedright
\bibliographystyle{abuser}
\bibliography{minorfreeflows}}


\newpage
\section*{Appendix: Small mimicking networks}

We show in this section that any three-terminal external flow network, and any four-terminal single-source external flow network, may be replaced by a mimicking network with at most one non-terminal vertex.

\begin{lemma}
\label{three-way}
Let $P$, $Q$, and $R$ be a partition of the terminals of any external flow network into three subsets.
Then $(P\cut Q\cup R)\le (P\cup Q\cut R)+(P\cup R\cut Q)$ and
$(P\cup Q\cut R)\le (P\cut Q\cup R)+(Q\cut P\cup R)$.
\end{lemma}

\begin{proof}
If $C_1$ and $C_2$ are minimum cuts separating $P\cup Q$ from $R$ and $P\cup R$ from $Q$, respectively, then $C_1\cup C_2$ separates $P$ from both $Q$ and $R$, and has capacity at most $(P\cup Q\cut R)+(P\cup R\cut Q)$; therefore, the minimum cut value $P\cut Q\cup R$ must be at most this capacity. The other inequality follows by a symmetric argument.
\end{proof}

\begin{lemma}[Chaudhuri et al.~\cite{cswz-cmn-00}]
\label{mimic3}
Any external flow network with three terminals has a mimicking network with four vertices and six edges.
\end{lemma}

\begin{proof}
Let the three terminals of flow network $G$ be $a$, $b$, and $c$, and form a network that consists of these three terminals together with a fourth vertex $d$, in the form of a star $K_{1,3}$ with six edges connecting $d$ in both directions to each terminal. For each terminal $q$, set the capacity of the edge from $q$ to $d$ to the minimum cut amount $q\cut(\{a,b,c\}\setminus\{q\})$ as measured in $G$, and set the capacity of the edge from $d$ to $q$ to be $(\{a,b,c\}\setminus\{q\})\cut q$.
It follows from Lemma~\ref{three-way} that each minimum cut separating one of the terminals from the other two consists of a single edge, and therefore that the capacity of this cut in the mimicking network exactly matches the capacity of the corresponding cut in $G$. The correctness of the mimicking network then follows by Lemma~\ref{hagerup}.
\end{proof}

It is tempting to try an even simpler network, a triangle with three vertices and six edges, in place of the four-vertex network of Lemma~\ref{mimic3}, but as Chaudhuri et al. observe this does not work: the triangle can only mimic networks for which $(a\cut bc)+(b\cut ac)+(c\cut ab)=(ab\cut c)+(ac\cut b)+(bc\cut a)$, and this equality is not necessarily true in other three-terminal external flow networks. However, in undirected flow networks, it is possible to form an undirected triangle network as a mimicking network, as shown in \cite{cswz-cmn-00}: the capacity of the edge from $a$ to $b$ should be $\frac12\bigl((a\cut bc)+(b\cut ac)-(c\cut ab)\bigr)$ and symmetrically for the other three edges.

We also need mimicking networks for four terminals, but only for single-source external flows.
In this case, the following variant of Lemma~\ref{hagerup} is helpful.

\begin{lemma}\label{ssmimic}
A single-source external flow $(x_1, \ldots, x_k)$ is realizable in a network $G=(V,E)$ with terminals $Q = \{ q_1, \ldots, q_k \}$ if and only if the following relations are satisfied:
\begin{eqnarray}
\sum_{i=1}^k x_i & = & 0\\
\sum_{q_i \in S} -x_i & \le & (x_1\cut S), \mbox{ for all nonempty } S \subseteq Q\setminus\{x_1\}.
\end{eqnarray}
\end{lemma}

\begin{proof}
Augment $G$ by a sink vertex $t$, add an edge of capacity $-x_i$ from each terminal $q_i$ ($i>1$) to $t$, and set the demand on $t$ to $x_1$. The result follows by the classical max-flow min-cut theorem: the desired single-terminal external flow exists if and only if the augmented graph has a two-terminal flow meeting the given demands, if and only if it has no cut with capacity less than $x_1$. But any cut in the augmented graph has a capacity equal to its capacity in the original graph plus the sum of the demands for the terminals on the source side of the cut, and each inequality in the statement of the lemma can be restated as requiring the capacities of some of these cuts to be at least $x_1$. So a cut with capacity less than $x_1$ exists, if and only if at least one of these inequalities is violated.
\end{proof}

\begin{lemma}
\label{four-way}
Let $P$, $Q$, $R$, and $S$ be four disjoint sets of terminals in an external flow network $G$. Then
$(P\cut Q\cup R\cup S)+(P\cup Q\cup R\cut S)\le (P\cup Q\cut R\cup S)+(P\cup R\cut Q\cup S)$.
\end{lemma}

\begin{proof}
Let $C_1$ and $C_2$ be minimum cuts separating $P\cup Q$ from $R\cup S$ and $P\cup R$ from $Q\cup S$, respectively. Let $A$ be the set of vertices of $G$ on the source side of $C_1$, and let $B$ be the set of vertices of $G$ on the source side of $C_2$. Then the edges of $C_1\cup C_2$ can be partitioned into three disjoint subsets: the edges from $A\cap B$ to $G\setminus(A\cap B)$, the edges from $A\cup B$ to $G\setminus (A\cup B)$, and any remaining edges that go in either direction between $A\setminus B$ and $B\setminus A$. The first of these three subsets of edges forms a cut separating $P$ from $Q\cup R\cup S$, and the second of these three subsets of edges forms a cut separating $P\cup Q\cup R$ from $S$, so the total capacity of $C_1\cup C_2$ is at least equal to $(P\cut Q\cup R\cup S)+(P\cup Q\cup R\cut S)$.
\end{proof}

\begin{lemma}
\label{mimic4}
Any single-source external flow network with four terminals has a mimicking network with five vertices and seven edges.
\end{lemma}

\begin{proof}
Given a single-source external flow network $G$, with terminals $s$, $a$, $b$, and $c$, we may assume without loss of generality (by permuting the final three terminals, if necessary) that $s\cut a$ is at least as large as $s\cut b$ and $s\cut c$, and that $s\cut ab$ is at least as large as $s\cut ac$. We create a mimicking network with a single nonterminal vertex $x$ and with the following seven edges:
\begin{itemize}
\item An edge from $s$ to $a$ with capacity $s\cut a$.
\item An edge from $s$ to $b$ with capacity $(s\cut ab)-(s\cut a)$.
\item An edge from $s$ to $c$ with capacity $(s\cut abc)-(s\cut ab)$.
\item An edge from $a$ to $x$ with capacity $(s\cut bc)+(s\cut a)-(s\cut abc)$.
\item An edge from $b$ to $x$ with capacity $(s\cut ac) + (s\cut ab) - (s\cut a) - (s\cut abc)$.
\item An edge from $x$ to $b$ with capacity $(s\cut b) + (s\cut a) - (s\cut ab)$.
\item An edge from $x$ to $c$ with capacity $(s\cut c) + (s\cut ab) - (s\cut abc)$.
\end{itemize}
Lemmas~\ref{three-way} and~\ref{four-way} can be used to show that all edge capacities in this network are non-negative.
As we now verify, each of the seven minimum cuts from $s$ to any subset of terminals in this network has the same capacity as the corresponding minimum cut in $G$:
\begin{itemize}
\item Every cut separating $s$ from $a$ must cut edge $sa$, so the cut that separates only that edge is the minimum cut, and has capacity $s\cut a$.
\item Every cut separating $s$ from $\{a,b\}$ must cut edges $sa$ and $sb$, so the cut that separates only those two edges is the minimum cut, and has capacity $s\cut ab$.
\item Every cut separating $s$ from $\{a,b,c\}$ must cut edges $sa$, $sb$, and $sc$, so the cut that separates only those three edges is the minimum cut, and has capacity $s\cut abc$.
\item There are three minimal cuts from $s$ to $\{b,c\}$: the cut $\{sa,sb,sc\}$, the cut $\{ax,sb,sc\}$, and the cut $\{xb,xc,sb,sc\}$. The first of these has capacity $(s\cut abc)\ge (s\cut bc)$, the second has capacity $s\cut bc$, and the third is at least as large as the second by Lemma~\ref{three-way}, so the minimum cut has capacity $s\cut bc$.
\item The minimal cuts from $s$ to $\{a,c\}$ are the cut $\{sa,sb,sc\}$ with capacity $s\cut abc$, the cut $\{sa,sc,bx\}$ with capacity $s\cut ac$, and the cut $\{sa,sc,xc\}$ which is at least as large as the previous cut by Lemma~\ref{three-way}. Therefore the minimum cut has capacity $s\cut ac$.
\item There are three minimal cuts from $s$ to $b$: the cut $\{sa,sb\}$, the cut $\{ax,sb\}$, and the cut $\{xb,sb\}$. The first of these has capacity $(s\cut ab)\ge (s\cut b)$, the second has capacity at least equal to $s\cut b$ by Lemma~\ref{four-way}, and the third has capacity $s\cut b$, so $s\cut b$ is the minimum cut value.
\item Among the cuts from $s$ to $c$, the ones that cut off another terminal from $s$ as well as $c$ fall into one of the previous cases, and therefore have capacity at least $s\cut c$. There are two remaining minimal cases: the cut $\{sc,ax,bx\}$ with capacity $(s\cut bc)-(s\cut abc)+(s\cut ac)$, at least as large as $s\cut c$ by Lemma~\ref{four-way}, and the cut $\{sc,xc\}$ with capacity $s\cut c$. Therefore, the minimum cut value is $s\cut c$.
\end{itemize}
Since all of these cuts have the same value as the corresponding cuts in $G$, it follows by Lemma~\ref{ssmimic} that the network constructed as above is a valid mimicking network for~$G$.
In the full paper we verify that all edge capacities are non-negative and that each of the seven minimum cuts from $s$ to a subset of terminals in this network has the same capacity as the corresponding cut in $G$, from which it follows that it is a valid mimicking network.
\end{proof}

The mimicking networks described in Lemmas~\ref{mimic3} and~\ref{mimic4} are shown in Figure~\ref{mimics}.

\newpage
\section*{Appendix: SPQR tree}

\begin{figure}[t]
\centering\includegraphics[width=1.5in]{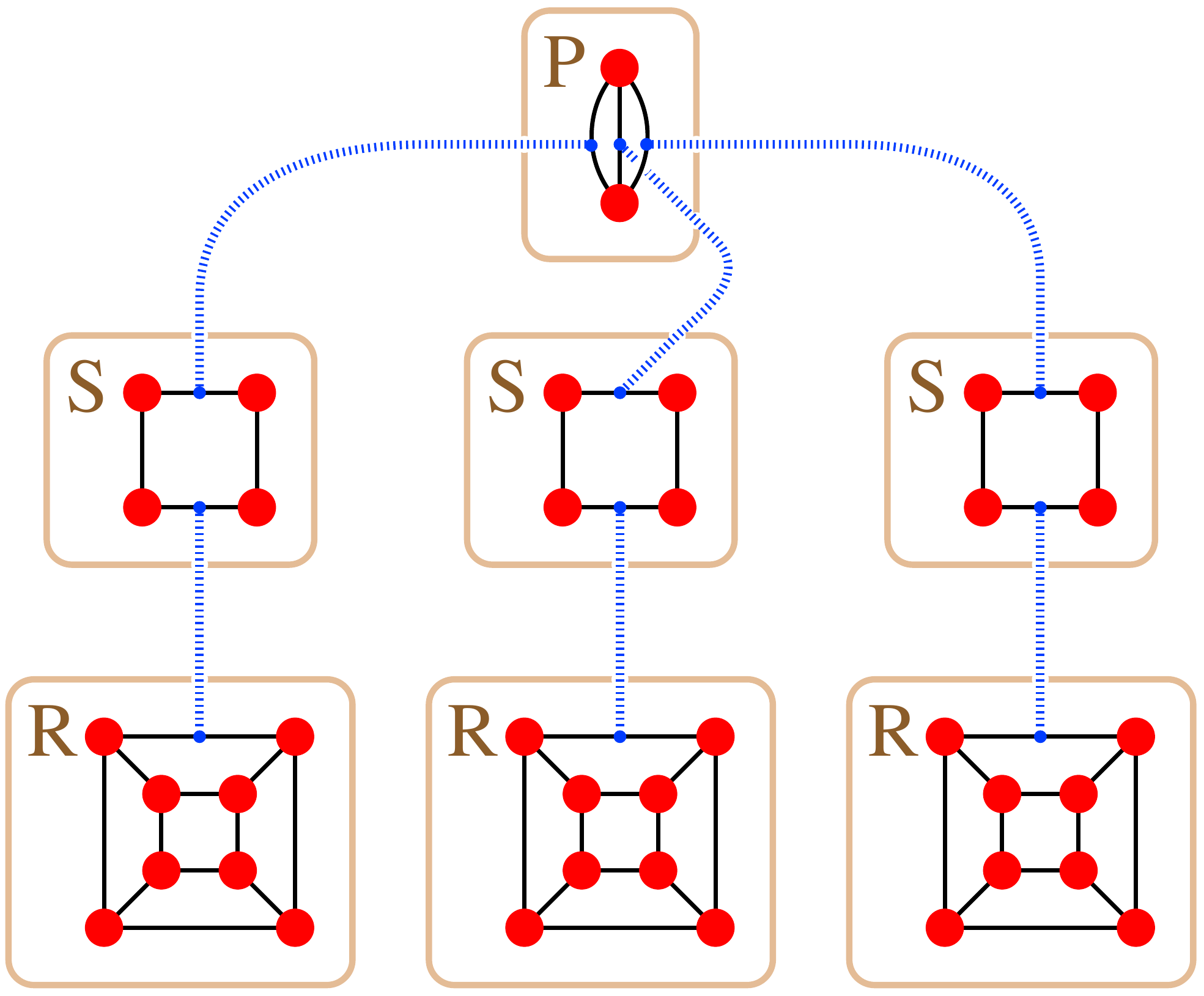}\qquad\includegraphics[width=1.5in]{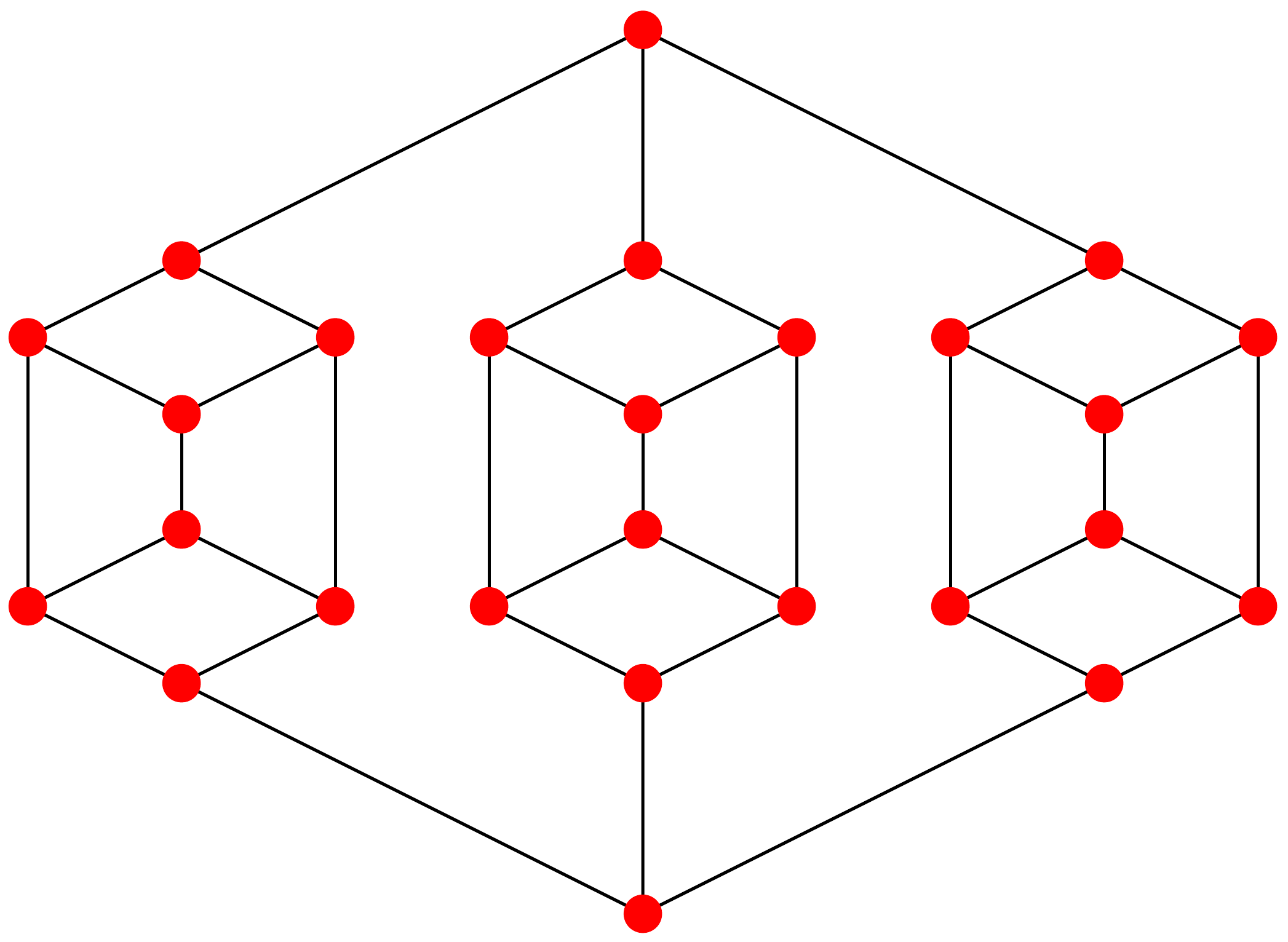}
\caption{An SPQR tree (left) and the graph that it represents (right). From the full version of~\cite{em-stop-10}.}
\label{fig:spqr}
\end{figure}

The SPQR tree of a graph is a formalized method of decomposing the graph into a 2-sum of simpler multigraphs, its \emph{triconnected components}~\cite{dt-ipt-89,ht-dgtc-73,m-scpcg-37}. In order to have a self-contained exposition, we describe it briefly here. The SPQR tree is properly defined only for 2-connected graphs; to apply it to graphs that may not be 2-connected, they need to be first decomposed into 2-connected components and the SPQR tree construction applied separately to each such component.

A \emph{cycle} is a 2-connected graph with $n$ vertices and $n$ edges; a \emph{bond} is a multigraph with two vertices and three or more edges, each of which has the same two endpoints.
Then an SPQR tree is a tree with the following properties:
\begin{itemize}
\item Each node $t_i$ of the SPQR tree is associated with a multigraph $G_i$ that is either a nontrivial 3-connected simple graph (an ``R'' node), a cycle (an ``S'' node), or a bond (a ``P'' node).
\item Each tree edge $t_it_j$ of the SPQR tree and is associated with a multigraph edge $e_{ij}$ in $G_i$ and another multigraph edge $e_{ji}$ in $G_j$. These two multigraph edges $e_{ij}$ and $e_{ji}$ are known as \emph{virtual edges}. 
\item Each multigraph edge can be a virtual edge for at most one SPQR tree edge.
\item Every SPQR tree edge connects two ``R'' nodes or two nodes of different types from each other.
\item Each ``P'' node has at most one non-virtual edge.
\end{itemize}
A tree of this type represents a graph $G$ that is formed by repeated 2-sum operations. Each 2-sum glues together the two designated virtual edges $e_{ij}$ and $e_{ji}$ for an SPQR tree edge and then deletes the glued-together edge. The order in which the 2-sums are performed does not affect the final result. Each edge that is not virtual, in the graph associated to each node, survives to become an edge in the overall graph. The graphs associated with the SPQR tree nodes are known as the \emph{triconnected components} of $G$.

Conversely, any 2-connected graph with two trivial exceptions, the graphs $K_1$ and $K_2$, has a unique representation as an SPQR tree; this representation can be constructed in linear time~\cite{dt-ipt-89,ht-dgtc-73}. In order to handle the two exceptions, we allow as a special case an SPQR tree with one ``Q'' node associated with one of these two exceptional graphs.

\end{document}